\documentclass[12pt]{article}

\def\colorful{0}

\usepackage{amsthm,amsfonts,amsmath,amssymb,epsfig,color,float,graphicx,verbatim}
\usepackage{multirow}

\newif\ifhyper\IfFileExists{hyperref.sty}{\hypertrue}{\hyperfalse}
\hypertrue
\ifhyper\usepackage{hyperref}\fi

\usepackage{enumitem}
\usepackage[capitalize]{cleveref} 

\usepackage{framed}
\usepackage{nicefrac}

\def\nnewcolor{1}
\ifnum\nnewcolor=1

\fi
\ifnum\nnewcolor=0

\fi

\ifnum\colorful=1
\newcommand{\new}[1]{{\color{red} #1}}

\else
\newcommand{\new}[1]{{#1}}

\fi

\newtheorem{theorem}{Theorem}

\newtheorem{lemma}[theorem]{Lemma}
\newtheorem{proposition}[theorem]{Proposition}

\newtheorem{claim}[theorem]{Claim}

\theoremstyle{definition}
\newtheorem{definition}[theorem]{Definition}

\newcommand{\R}{\mathbb{R}}

\newcommand{\Z}{\mathbb{Z}}

\newcommand{\opt}{\mathrm{OPT}}
\newcommand{\poly}{\mathrm{poly}}
\newcommand{\polylog}{\mathrm{polylog}}

\newcommand{\dtv}{d_{\mathrm TV}}

\newcommand{\ignore}[1]{}

\newcommand{\eps}{\epsilon}

\newcommand{\eqdef}{\stackrel{{\mathrm {\footnotesize def}}}{=}}

\newenvironment{algorithm}[1][\  ] %
{ \rm
\begin{tabbing}
....\=.....\=.....\=.....\=.....\=  \+ \kill
} %
{\end{tabbing} }

\floatstyle{ruled}
\newfloat{Algorithm}{H}{alg}
\newfloat{Subroutine}{H}{sub}

{
\begin{minipage}{1.0\linewidth} \begin{algorithm} %
} { \end{algorithm} \end{minipage} }

\author{
Ilias Diakonikolas\thanks{Part of this work was performed while the author was at the University of Edinburgh.
Supported in part by EPSRC grant EP/L021749/1 and a Marie Curie Career Integration grant.}\\
University of Southern California\\
{\tt diakonik@usc.edu}\\
\and
Daniel M. Kane\thanks{Supported in part by NSF Award CCF-1553288 (CAREER).
Some of this work was performed while visiting the University of Edinburgh.}\\
University of California, San Diego\\
{\tt dakane@cs.ucsd.edu}\\
\and
Alistair Stewart\thanks{Part of this work was performed while the author was at the University of Edinburgh.
Supported by EPSRC grant EP/L021749/1.}\\ University of Southern California\\
{\tt alistais@usc.edu}
}

\title{Efficient Robust Proper Learning of Log-concave Distributions}

\begin{document}

\maketitle

\begin{abstract}
We study the {\em robust proper learning} of univariate log-concave distributions
(over continuous and discrete domains).
Given a set of samples drawn from an unknown target distribution, we want to compute a log-concave
hypothesis distribution that is as close as possible to the target, in total variation distance.
In this work, we give the first computationally efficient algorithm
for this learning problem. Our algorithm achieves the information-theoretically optimal
sample size (up to a constant factor), runs in polynomial time,
and is robust to model misspecification with nearly-optimal
error guarantees.

Specifically, we give an algorithm that,
on input $n=O(1/\eps^{5/2})$ samples from an unknown distribution $f$,
runs in time $\widetilde{O}(n^{8/5})$,
and outputs a log-concave hypothesis $h$ that (with high probability) satisfies
$\dtv(h, f) = O(\opt)+\eps$, where $\opt$ is the minimum total variation distance between $f$
and the class of log-concave distributions.
\new{Our approach to the robust proper learning problem is quite flexible and may be applicable
to many other univariate distribution families.}
\end{abstract}

\section{Introduction}  \label{sec:intro}

\subsection{Background and Motivation} \label{ssec:background}

Suppose that we are given a number of samples drawn from an unknown target distribution that belongs to (or is well-approximated by)
a given family of distributions ${\cal D}$. Our goal is to approximately estimate (learn) the target distribution in a precise way.
Estimating a distribution from samples is a fundamental unsupervised learning problem
that has been studied in statistics since the late nineteenth century~\cite{Pearson}.
During the past couple of decades, there has been a large body of work
in computer science on this topic with a focus on computational efficiency~\cite{KMR+:94short}.

The performance of a distribution learning (density estimation) algorithm is typically evaluated by the following criteria:
\begin{itemize}
\item {\em Sample Complexity:} For a given error tolerance, the algorithm should require a small number of samples,
ideally matching the information-theoretic minimum.
\item {\em Computational Complexity:}  The algorithm should run in time polynomial
in the number of samples provided as input.

\item {\em Robustness:} The algorithm should provide error guarantees under model misspecification,
i.e., even if the target distribution does not belong in the target family ${\cal D}$.
The goal here is to be competitive with the best approximation of the unknown distribution
by {\em any} distribution in ${\cal D}$.
\end{itemize}
In {\em non-proper} learning, the goal of the learning algorithm is to output
an approximation to the target distribution without any constraints on its representation.
In {\em proper} learning, we require in addition that
 the hypothesis is a member of the family ${\cal D}$.
Note that these two notions of learning are essentially equivalent in terms of sample complexity
(given any accurate hypothesis, we can do a brute-force search to find its closest distribution in ${\cal D}$),
but not necessarily equivalent in terms of computational complexity.

In many learning situations it is desirable to compute a proper hypothesis,
i.e., one that belongs to the underlying family ${\cal D}$. A proper hypothesis is
usually preferable due to its interpretability. In particular,
a practitioner may not want to use a density estimate, unless it is proper.
For example, one may want the estimate to have the properties of the underlying family,
either because this reflects some physical understanding of the inference problem,
or because one might only be using the density estimate as the first stage of a more involved procedure.

The aforementioned discussion raises the following algorithmic question:
{\em Can one obtain a {\em proper} learning algorithm for a given distribution family ${\cal D}$
whose running time matches that of the {\em best non-proper} algorithm for ${\cal D}$?}
Perhaps surprisingly, our understanding of this natural question remains quite poor.
In particular, little is known about the complexity of proper learning
in the unsupervised setting of learning probability distributions.
In contrast, the computational complexity of proper learning
has been extensively investigated in the supervised setting of PAC learning
Boolean functions~\cite{KearnsVazirani:94, Feldman15}, with several algorithmic and computational
intractability results obtained in the past decades.

In this work, we study the problem of {\em robust proper learning} for the family of
univariate log-concave distributions (over $\R$ or $\Z$) (see Section~\ref{ssec:results} for a precise definition).
Log-concave distributions constitute a broad non-parametric family
that is very useful for modeling and inference~\cite{Walther09}.
In the discrete setting, log-concave distributions encompass a range of fundamental
types of discrete distributions, including binomial, negative binomial,
geometric, hypergeometric, Poisson, Poisson Binomial, hyper-Poisson,
P\'{o}lya-Eggenberger, and Skellam distributions (see Section~1 of
\cite{BJR11}).  In the continuous setting, they include
uniform, normal, exponential, logistic, extreme value,
Laplace, Weibull, Gamma, Chi and Chi-Squared and Beta distributions (see~\cite{BagnoliBergstrom05}).
Log-concave distributions have been studied in a wide range of different contexts including
economics \cite{An:95}, statistics and probability theory (see~\cite{SW14-survey} for a recent survey),
theoretical computer science~\cite{LV07},
and algebra, combinatorics and geometry~\cite{Stanley:89}.

\subsection{Our Results and Comparison to Prior Work}  \label{ssec:results}

The problem of density estimation for log-concave distributions is of central importance
in the area of non-parametric shape constrained inference. As such,
this problem has received significant attention in the statistics literature,
see~\cite{Cule10a, DumbgenRufibach:09, DossW16, ChenSam13, KimSam14, BalDoss14, HW16} and references therein,
and, more recently, in theoretical computer science~\cite{CDSS13,
CDSS14, ADLS15, AcharyaDK15, CanonneDGR16, DKS16lcd}.
In this section, we state our results and provide a brief comparison to the most relevant prior work.
See Section~\ref{sec:related} for a a more detailed summary of related work.

We study univariate log-concave distributions over both continuous and discrete domains.

\begin{definition} \label{def:lc}
A function $f: \R \to \R_+$ with respect to Lebesgue measure is log-concave
if $f = \exp(\phi)$ where $\phi : \R \to [-\infty, \infty)$ is a concave function.
A function $f: \Z \to [0, 1]$ is log-concave if $f^2(x) \geq f(x-1) \cdot f(x+1)$ for all $x \in \Z$
and $f$ has no internal zeroes. We will denote by $\mathcal{LC}(D)$ the family of log-concave
densities over $D$.
\end{definition}

We use the following notion of agnostic learning under the total variation distance, denoted by $\dtv$:

\begin{definition}[Agnostic Proper Learning]  \label{def:learning}
Let ${\cal D}$ be a family of probability density functions on domain $D$.
A randomized algorithm $A^{\cal D}$ is an {\em agnostic distribution learning algorithm for $\cal D,$}
if for any $\eps>0,$ and any probability density function $f: D \to \R_+$, on input $\eps$ and sample access to $f$,
with probability $9/10,$ algorithm $A^{\cal D}$
outputs a hypothesis density $h \in \mathcal{D}$ such that
$\dtv(h, f) \leq O(\opt) + \eps$,
where $\opt \eqdef \inf_{g \in \mathcal{D}} \dtv(f, g)$.
\end{definition}

Given the above terminology, we can state our main algorithmic result:

\begin{theorem}[Main Result] \label{thm:main-final}
There exists an algorithm that, given $n = O(\eps^{-5/2})$ samples from an arbitrary density $f: D \to \R_+$,
where $D = \R$ or $D = \Z$,
runs in time $\widetilde{O}(n^{8/5})$ and outputs a hypothesis $h \in \mathcal{LC}(D)$ such that
with probability at least $9/10$ it holds
$\dtv(h, f) \leq O(\opt) + \eps$,
where $\opt \eqdef \inf_{g \in  \mathcal{LC}(D)} \dtv(f, g)$.
\end{theorem}

We note that the sample complexity of our algorithm is optimal (up to constant factors),
as follows from previous work~\cite{DL:01, CDSS13}.
Our algorithm of Theorem~\ref{thm:main-final} is the first polynomial time
{\em agnostic proper} learning algorithm for the family of log-concave distributions. In particular,
previous polynomial time learning algorithms for log-concave distributions
were either {\em non-proper}~\cite{CDSS13, CDSS14, ADLS15} or {\em non-agnostic}~\cite{AcharyaDK15, CanonneDGR16}.
Specifically, the sequence of works~\cite{CDSS13, CDSS14, ADLS15} give computationally efficient agnostic learning algorithms
that are inherently non-proper. Two recent works \cite{AcharyaDK15, CanonneDGR16} give proper learning algorithms
for discrete log-concave distributions that are provably non-agnostic.
It should be noted that the sample complexity and running time of the non-robust proper algorithms
in~\cite{AcharyaDK15, CanonneDGR16} are significantly worse than ours. We elaborate on this point
in the following subsection.

\subsection{Related Work} \label{sec:related}

\paragraph{Distribution Learning.} 

Distribution learning is a paradigmatic inference problem
with an extensive literature in statistics (see, e.g., the books~\cite{BBBB:72, DG85, Silverman:86,Scott:92,DL:01}).
A number of works in the statistics community
have proposed proper estimators (relying on a maximum likelihood approach)
for various distribution families. Alas, typically, these estimators are either intractable
or their computational complexity is not analyzed.

A body of work in theoretical computer science has focused on
distribution learning from a computational complexity perspective; see, e.g.,
\cite{KMR+:94short,FreundMansour:99short, AroraKannan:01, CGG:02, VempalaWang:02,FOS:05focsshort, BelkinSinha:10, KMV:10,
DDS12soda, DDS12stoc, DDOST13focs, CDSS13, CDSS14, CDSS14b, ADLS15, DKS15, DaskalakisDKT15, DiakonikolasKS15a}.
We note that, while the majority of the literature studies either non-proper learning or
parameter estimation, proper learning algorithms have been obtained for a number of families,
including mixtures of simple parametric models~\cite{FOS:05focsshort, DK14, SOAJ14, LiS15a},
and, Poisson binomial distributions~\cite{DKS15b}.

\paragraph{Prior Work on Learning Log-concave Distributions.} 

Density estimation of log-concave distributions has been extensively investigated
in the statistics literature~\cite{DumbgenRufibach:09, GW09sc, Walther09, DossW16, BJRP13, ChenSam13, KimSam14, BalDoss14}
with a focus on analyzing the maximum likelihood estimator (MLE).
For the {\em continuous} case, the sample complexity of the problem has been characterized~\cite{DL:01},
and it is known~\cite{KimSam14, HW16} that the MLE is sample efficient.
It has been shown~\cite{DR11stat} that the MLE for continuous log-concave densities c
an be formulated as a convex program,
but no explicit upper bound on its running time is known.
We remark here that the MLE is known to be non-agnostic with respect to the total variation distance, 
even for very simple settings (e.g., for Gaussian distributions).

Recent work in theoretical computer science~\cite{CDSS13, CDSS14, ADLS15} gives sample-optimal, agnostic,
and computationally efficient algorithms for learning log-concave distributions (both continuous and discrete).
Alas, all of these algorithms are {\em non-proper}, i.e.,
they output a hypothesis that is not log-concave. For the case of {\em discrete} log-concave distributions supported on $[n]$,
two recent papers~\cite{AcharyaDK15, CanonneDGR16} obtain proper algorithms
that use $\poly(1/\eps)$ samples and run in time $\poly(n/\eps)$.
Roughly speaking,~\cite{AcharyaDK15, CanonneDGR16}
proceed by formulating the proper learning problem as a convex program.

Here we would like to emphasize three important differences between~\cite{AcharyaDK15, CanonneDGR16}
and the guarantees of Theorem~\ref{thm:main-final}.
First, the algorithms of~\cite{AcharyaDK15, CanonneDGR16} are {\em inherently non-agnostic}.
Second, their sample complexity is sub-optimal, namely $\Omega(1/\eps^5)$, while our algorithm is sample-optimal.
Third, the linear programming formulation that they employ has size (i.e., number of variables and constraints) $\Omega(n)$,
i.e., its size depends on the support of the underlying distribution.
As a consequence, the runtime of this approach is prohibitively slow, for large $n$.
In sharp contrast, our algorithm's running time is independent of the support size,
and scales sub-quadratically with the number of samples.

\subsection{Overview of our Techniques} \label{sec:techniques}

In this section, we provide a high-level overview of our techniques.
Our approach to the proper learning problem is as follows:
Starting with an accurate non-proper hypothesis, we fit a log-concave density to this hypothesis.
This fitting problem can be formulated as a (non-convex) discrete optimization problem that we can solve efficiently
by a combination of structural approximation results and dynamic programming.
\new{Specifically, we are able to phrase this optimization problem 
as a shortest path computation in an appropriately defined
edge-weighted directed acyclic graph.}

In more detail, our agnostic proper learning algorithm works in two steps:
First, we compute an accurate {\em non-proper} hypothesis, $g$, by applying any 
efficient non-proper agnostic learning algorithm as a black-box (e.g.,~\cite{CDSS14, ADLS15}).
In particular, we will use the non-proper learning algorithm of~\cite{ADLS15} that outputs 
a piecewise linear hypothesis distribution $g$.
To establish the sample-optimality of the~\cite{ADLS15} algorithm, 
one requires the following structural result that we establish (Theorem~\ref{approxThrm}):
Any log-concave distribution (continuous or discrete) can be $\eps$-approximated, in total variation distance,
by a piecewise linear distribution with $O(\epsilon^{-1/2})$ interval pieces.
Since $\Omega(\epsilon^{-1/2})$ interval pieces are required for such an approximation,
our bound on the number of intervals is tight. It should be noted that a quantitatively similar
structural result was shown in~\cite{CDSS14} for {\em continuous} log-concave distributions,
with a bound on the number of pieces that is sub-optimal up to logarithmic factors. For the discrete case,
no such structural result was previously known.


Since $g$ is not guaranteed to be log-concave, our main algorithmic step efficiently post-processes
$g$ to compute a log-concave distribution that is (essentially) as close to $g$ as possible, in total variation distance.
To achieve this, we prove a new structural result (Lemma~\ref{lem:struct}) 
showing that the \new{closest} log-concave distribution
can be well-approximated by a log-concave {\em piecewise exponential} distribution
whose pieces are determined only by the mean and standard deviation of $g$.
Furthermore, we show (Proposition~\ref{prop:this-form-close}) 
we can assume that the values of this approximation
at the breakpoints can be appropriately discretized. These structural
results are crucial for our algorithmic step outlined below.

From this point on, our algorithm
proceeds via dynamic programming. Roughly speaking, we record the best possible error
in approximating $g$ by a function of the aforementioned form on the interval $(-\infty,x]$ for various
values of $x$ and for given values of $h(x), h'(x)$. Since knowing $h(x)$ and $h'(x)$
is all that we need in order to ensure that the rest of the function is log-concave,
this is sufficient for our purposes. It turns out that this dynamic program
can be expressed as a shortest path computation in a graph that we construct.
The time needed to compute the edge weights of this graph depends on the 
description of the non-proper hypothesis $g$. In our case, $g$ is a piecewise linear distribution
and all these computations are manageable.

\subsection{Organization}
In Section~\ref{sec:prelims} we record the basic probabilistic ingredients we will require.
In Section~\ref{sec:results} we prove our main result.
Finally, we conclude with a few open problems in Section~\ref{sec:concls}.

\section{Preliminaries} \label{sec:prelims}

For $n \in \Z_+$, we denote $[n] \eqdef \{1,\dots,n\}$.
For $u \in \R$, we will denote $\exp(u) \eqdef e^u$.
Let $f: \R \to \R$ be a Lebesgue measurable function.
We will use $f(A)$ to denote $\int_{x \in A} f(x) dx$.

A Lebesgue measurable function $f: \R \to \R$ is a probability density function (pdf)
if $f(x) \geq 0$ for all $x \in \R$ and  $\int_{\R} f(x) dx = 1$.
We say that $f: \R \to \R$ is a pseudo-distribution if $f(x) \geq 0$ for all $x \in \R.$
A function $f: \Z \to \R$ is a probability mass function (pmf)
if if $f(x) \geq 0$ for all $x \in \Z$ and  $\sum_{x \in \Z} f(x) = 1$.
We will similarly use $f(A)$ to denote $\sum_{x \in A} f(x)$.
We say that $f: \Z \to \R$ is a pseudo-distribution if $f(x) \geq 0$ for all $x \in \Z.$

For uniformity of the exposition, we will typically use $D$ to denote the domain of our functions,
where $D$ is $\R$ for the continuous case and $\Z$ for the discrete case.
We will use the term {\em density} to refer to either a pdf or pmf.

The $L_1$-distance between $f, g: D \to \R$ over $I \subseteq D$, denoted
$\|f-g\|_{1}^{I}$, is  $\int_{I} |f(x) - g(x)| dx$ for $D = \R$ and
$\sum_{x \in I} |f(x) - g(x)|$ for $D = \Z$; when $I = D$ we suppress the superscript $I$.
The {\em total variation distance} between densities $f, g: D \to \R_+$ is defined as
$\dtv\left(f, g \right) \eqdef (1/2) \cdot \| f -g  \|_1$.

Our algorithmic and structural results make essential use of continuous
piecewise exponential functions, that we now define:

\begin{definition} \label{def:pe}
Let $I = [\alpha, \beta] \subseteq D$, where $\alpha , \beta \in D$.
A function $g: I \to \R_+$ is {\em continuous $k$-piecewise exponential} if there exist
$\alpha \equiv x_1 < x_2 <  \ldots < x_{k} < x_{k+1} \equiv \beta$, $x_i \in D$,
such that for all $i \in [k]$ and $x \in I_i \eqdef [x_i, x_{i+1}]$ we have that $g(x) = g_i(x)$,
where  $g_i(x) \eqdef \exp(c_i x + d_i)$, $c_i, d_i \in \R$.
\end{definition}

Note that the above definition implies that $g_i(x_{i+1}) = g_{i+1}(x_{i+1})$, for all $i \in [k]$.

We will also require a number of useful properties of log-concave densities,
summarized in the following lemma:
\begin{lemma} \label{lem:lc-facts}
Let $f: D \to \R_+$ be a log-concave density
with mean $\mu$ and standard deviation $\sigma$.
Then: (i) If $D=\R$ or $D=\Z$ \new{and $\sigma$ is at least a sufficiently large constant,}
$1/(8\sigma) \leq M_f \eqdef \max_{x \in D} f(x) \leq 1/\sigma$, and
(ii) \new{$f(x) \leq \exp\left(1-|x-\mu| M_f/e\right)M_f$} for all $x \in D$.
\end{lemma}
For the case of continuous log-concave densities, (i) appears as Lemma~5.5 in~\cite{LV07}, 
and the discrete case follows similarly.
To show (ii) \new{we note that, since $f$ is unimodal, $f(\mu+e/M_f)$ and $f(\mu-e/M_f)$ 
are each at most $M_f/e$. The claim then follows from log-concavity.}

\begin{lemma} \label{lem:robust-moments}
Let $f: D \to \R_+$ be a log-concave density with mean $\mu$ and standard deviation $\sigma$. 
\new{Assume that either $D=\R$ or that $\sigma$ is sufficiently large.}
Let $g: D \to \R_+$ be a density with $\dtv(f, g) \leq 1/10$.
Given an explicit description of $g$, we can efficiently compute values $\tilde \mu$ and $\tilde \sigma$ so that
$|\mu-\tilde\mu| \leq 2\sigma$ and $3\sigma/10 \leq \tilde \sigma \leq 6\sigma.$
\end{lemma}

The proof of the above lemma uses the log-concavity of $f$ and is deferred to Appendix \ref{sec:app-lemma2}.

\section{Proof of Theorem~\ref{thm:main-final}: Our Algorithm and its Analysis} \label{sec:results}

\subsection{Approximating Log-concave Densities by Piecewise Exponentials} \label{ssec:structural}
Our algorithmic approach relies on approximating log-concave densities by continuous piecewise exponential functions.
Our first structural lemma states that we can approximate a log-concave density
by a continuous piecewise exponential pseudo-distribution with an appropriately small set of interval pieces.
\begin{lemma} \label{lem:struct}
Let $f: D \to \R_+$ be a log-concave density with mean $\mu$, standard deviation $\sigma$ 
\new{at least a sufficiently large constant}, and
$\epsilon>0$. Let $I = [\alpha, \beta] \subseteq D$ be such that
$\alpha<\mu<\beta$ and $|\alpha-\mu|, |\beta  -\mu|  = \Theta(\log(1/\epsilon)\sigma)$,
with the implied constant sufficiently large.
Let $k$ be an integer so that either $k=\new{\Theta}(\log(1/\eps)/\eps)$, 
or $k=\beta-\alpha = \new{O(\log(1/\eps)/\eps)}$ and $D=\Z$.
Consider the set of equally spaced endpoints
$\alpha \equiv x_1 < x_2 <  \ldots < x_{k} < x_{k+1} \equiv \beta$.
There exist indices $1 \leq l < r \leq k+1$ and a \new{log-concave}
continuous piecewise exponential pseudo-distribution
$g:  I \to \R_+$ with $\|f - g\|_1 \leq O(\eps)$ such that the following are satisfied:
\begin{itemize}
\item[(i)] $g(x) = 0$, for all $x \not \in J \eqdef [x_l, x_r]$.

\item[(ii)] For all $l\leq i \leq r$ it holds $g(x_i) = f(x_i)$.

\item[(iii)] For $l\leq i < r$, $g$ is exponential on $[x_i,x_{i+1}]$.
\end{itemize}
\end{lemma}

\begin{proof}
\new{If $D=\Z$ and $\sigma$ is less than a sufficiently small constant, 
then Lemma \ref{lem:lc-facts} implies that all but an $\epsilon$-fraction of the mass of $f$ 
is supported on an interval of length $O(\log(1/\epsilon))$. 
If we let $I$ be this interval and take $k=|I|$, 
we can ensure that $g=f$ on $I$ and our result follows trivially. 
Henceforth, we will assume that either $D=\R$ or that $\sigma$ is sufficiently large.}

The following tail bound is a consequence of Lemma ~\ref{lem:lc-facts}
and is proved in Appendix \ref{sec:app-claim1}:
\begin{claim} \label{clm:tail-bound}
Let $f: D \to \R_+$ be a log-concave density with mean $\mu$ and standard deviation $\sigma$.
Let $\alpha \leq \mu - \Omega(\sigma(1+ \log(1/\eps)))$ and
$\beta \geq \mu + \Omega(\sigma(1+ \log(1/\eps))))$.
Then, $\|f\|_1^{(-\infty,\alpha)} \leq \eps$ and $\|f\|_1^{(\beta,\infty)} \leq \eps.$
\end{claim}

By Claim~\ref{clm:tail-bound}, it suffices to exhibit the existence of the function $g: I \to \R_+$
and show that $\|f-g\|_1^{I} \leq O(\eps).$
We note that if $D=\Z$ and $k=\beta-\alpha$,
then we may take $g=f$ on $I$, and this follows immediately.
\new{Hence, it suffices to assume that $k>C(\log(1/\eps)/\eps)$, 
for some appropriately large constant $C>0$.}

Next, we determine appropriate values of $l$ and $r$.
In particular, we let $l$ be the minimum value of $i$ so that $f(x_i)>0$,
and let $r$ be the maximum. We note that the probability measure
of $\mathrm{supp}(f) \setminus J$ is at most
$O(|\beta-\alpha|/k)=O(\epsilon \sigma)$.
Since \new{$M_f=O(1/\sigma)$, by Lemma~\ref{lem:lc-facts}(i), } 
we have that $f(I \setminus J)=O(\epsilon)$.
Therefore, it suffices to show that $\|f-g\|^J_1 = O(\eps)$.

We take $k = \Omega(\log(1/\eps)/\eps)$.
Since the endpoints $x_1, \ldots, x_{k+1}$ are equally spaced,
it follows that $L =  |x_{i+1} - x_i| = O(\eps \sigma)$.

For $l\leq i < r$, for $x\in [x_i,x_{i+1}]$ we let $g(x)$ be given
by the unique exponential function that interpolates $f(x_i)$ and $f(x_{i+1})$.
We note that this $g$ clearly satisfies properties (i), (ii), and (iii).
It remains to show that $\|f-g\|_1^J=O(\eps).$

Let $I_j = [x_j, x_{j+1}]$ be an interval containing a mode of $f$.
We claim that $\|f - g\|^{I_j}_1 = O(\eps)$.
This is deduced from the fact that $\max_x g(x) \leq \max_x f(x) \leq 1/\sigma$,
where the first inequality is by the definition of $g$ and the second follows by Lemma~\ref{lem:lc-facts} (i).
This in turn implies that the probability mass of both $f(I_j)$ and $g(I_j)$
is at most $L \cdot (1/\sigma) = O(\eps)$.

We now bound from above the contribution to the error
coming from the intervals $I_1, \ldots, I_{j-1}$, 
i.e., the quantity $\sum_{i=1}^{j-1} \|f-g\|_1^{I_i}$.
Since all $I_i$'s have length $L$, and $f, g$ are monotone non-decreasing agreeing on the endpoints,
we have that the aforementioned error term is at most
$$L \cdot \sum_{i=1}^{j-1} \left( f(x_{i+1})-f(x_{i}) \right) = O(\epsilon \sigma) \cdot M_f
\leq O(\epsilon \sigma) \cdot (1/\sigma) = O(\epsilon) \;.$$
A symmetric argument shows the error coming from the
intervals $I_{j+1}, \ldots, I_{k}$ is also $O(\eps)$.
An application of the triangle inequality completes the proof.
\end{proof}

The following proposition establishes the fact that the log-concave 
piecewise exponential approximation can be assumed to be appropriately discretized:

\begin{proposition} \label{prop:this-form-close}
Let $f: D \to \R_+$ be a log-concave density with mean $\mu$, standard deviation $\sigma$
\new{at least a sufficiently large constant}, and $\epsilon>0$.
Let $\tilde \sigma = \Theta(\sigma)$.
Let $I = [\alpha, \beta] \subseteq D$ containing $\mu$ be such that $|\alpha-\mu|, |\beta-\mu|  = \Theta (\log(1/\epsilon) \tilde \sigma)$,
where the implied constant is sufficiently large.
Consider a set of equally spaced endpoints $\alpha \equiv x_1 < x_2 <  \ldots < x_{k} < x_{k+1} \equiv \beta$, $x_i \in D$,
where either $k = \new{\Theta}(\log(1/\eps)/\eps)$, or $D=\Z$ and $k=\beta-\alpha \new{= O(\log(1/\eps)/\eps)}$.
There exist indices $1 \leq l < r \leq k+1$ and a \new{log-concave} continuous piecewise exponential pseudo-distribution
$h:  I \to \R_+$ with $\|f - h\|_1 \leq O(\eps)$ such that the following are satisfied:
\begin{itemize}

\item[(i)] $h(x) > 0$ if and only if $x \in J \eqdef [x_l, x_r]$.

\item[(ii)] For each endpoint $x_i$, $i \in [l, r]$,
we have that (a) $\log(h(x_i)\tilde \sigma)$ is an integer multiple of $\eps/k$,
and (b) $\new{-O(\log(1/\eps)) \leq \log(h(x_i)\tilde \sigma) \leq O(1) }$.

\item[(iii)] For any $i \in [l, r-1]$ we have
$|\log\left(h(x_i)\tilde \sigma\right)-\log\left(h(x_{i+1})\tilde \sigma\right)|$ is of the form $b \cdot \eps \cdot  2^c/k$,
for integers $|b| \leq (1/\eps) \log(1/\epsilon)$ and $0 \leq c \leq O(\log (1/\eps))$.
\end{itemize}
\end{proposition}
\begin{proof}
Let $g: I \to \R_+$ be the pseudo-distribution given by the Lemma~\ref{lem:struct}.
We will construct our function $h: I \to \R_+$ such that $\|h-g\|_1  = O(\eps)$.
For notational convenience, for the rest of this proof we will denote
$a_i \eqdef \log \left(h(x_i)\tilde \sigma \right)$ and  $a'_i \eqdef \log\left(g(x_i)\tilde \sigma\right)$ for $i \in [k+1]$.

We define the function $h$ to be supported on the interval $J = [x_l, x_r]$
specified as follows: The point $x_l$ is the leftmost endpoint such that $g(x_l) \ge \eps^2 \new{/ \tilde \sigma}$
or equivalently  $a'_l  \ge -2\log(1/\eps)$. Similarly, the point $x_r$ is the rightmost endpoint
such that $g(x_r) \ge \eps^2 \new{/ \tilde \sigma}$ or equivalently  $a'_r  \ge -2\log(1/\eps)$.

We start by showing that the probability mass of $g$ outside the interval $J$ is $O(\eps)$.
This is because $g(x)\leq \eps^2/\tilde \sigma$ off of $J$, and so has total mass at most $\eps^2/\tilde \sigma (\beta-\alpha)= O(\eps)$.

To complete the proof, we need to appropriately define $h$
so that it satisfies conditions (ii) and (iii) of the proposition statement,
and in addition that $\|h-g\|_1^{J} \leq O(\eps)$.

We note that $-O(\log(1/\eps)) \leq a'_i \leq O(1)$ for all $i \in [l,r]$.
Indeed, since $a'_l, a'_r  \ge -2\log(1/\eps)$ and $g$ is log-concave,
we have that $a'_i \ge -\log(1/\eps)$ for all $i \in [l,r]$.
Also, since $a'_i = \log\left(g(x_i) \tilde \sigma\right) =\log\left(f(x_i) \tilde \sigma\right)$,
we obtain $a'_i \leq \log (M_f \tilde \sigma) \leq \log 6 \leq 1.8$.
We will construct $h$ so that $|a_i'-a_i|=O(\eps)$ for all $l\leq i \leq r$.
We claim that this is sufficient since it would imply that $\log(g(x)/h(x))=O(\eps)$ for all $x\in J$.
This in turn implies that $h(x)=g(x)+O(\eps g(x))$, and thus that $\|h-g\|_1^J = \int_J O(\eps g(x)) dx = O(\eps)$.

We are now ready to define $a_i$, $i \in [l, r]$.
Let $j$ be such that $x_j$ is a mode of $g$.
Let $d_i$ be obtained by rounding $d'_i \eqdef a'_i-a'_{i-1}$ as follows:
Let $c_i$ be the least non-negative integer such that
$|d'_i|  \leq 2^{c_i} \new{\log(1/\eps)/k}$.
Then, we define $d_i$ to be $d'_i$ rounded
to the nearest integer multiple of $2^{c_i}\eps/k$
(rounding towards $0$ in the case of ties).
Let $a_j$ be the nearest multiple of $(\eps/k)$ to $a'_j$.
Let $a_i = a_j + \sum_{k=j+1}^i d_k$ for $j > i$,
and $a_i= a_j- \sum_{k=i+1}^j d_k$ for $i < j$.
We define $h$ to be the continuous piecewise exponential
function with $h(x_i)=\exp(a_i)/\tilde \sigma$, $i \in [l, r]$,
that is exponential on each of the intervals $[x_i, x_i+1]$,
for $i \in [l, r-1]$.

By construction, for all $i \in [l, r]$,
$a_i$ is an integer multiple of $\eps/k$
and $|a_i-a_{i+1}|$ is of the form $b \cdot \eps \cdot 2^c/k$ for integers
$0 \leq b \leq (1/\eps)\log(1/\eps)$, and $0 \leq c \leq O(\log 1/\eps)$.
Since $g$ is log-concave, we have $a'_i-a'_{i-1} \geq a'_{i+1}-a'_i$.
Note that the rounding of the $d_i$'s is given by a monotone function
and thus we also have $a_i-a_{i-1} \geq a_{i+1}-a_i$. Hence,
$h$ is also log-concave.
Since $|a'_i| \leq \log(1/\eps)$, $i \in [l, r]$, the definition of the $c_i$'s yields
$\sum_{i} 2^{c_i} \new{\log(1/\eps)/k}  =O( \log(1/\eps))$
or $\sum_{i} 2^{c_i} =O(k)$.
Since $|d'_i-d_i| \leq 2^{c_i} \eps/k$,
we have that $|a_i-a'_i| \leq (\eps/k) + \sum_i 2^{c_i}\eps/k \leq  O(\eps)$. This completes the proof.
\end{proof}

\subsection{Main Algorithm} \label{ssec:DP}

\begin{theorem} \label{thm:main}
Let $g: D \to \R_+$ be a density and let $\opt = \inf_{f \in \mathcal{LC}(D)} \dtv(g, f)$.
There exists an algorithm that, given $g$ and $\eps>0$, outputs
an explicit log-concave density $h$ such that $\dtv(g, h)=O(\opt+ \eps)$.
The algorithm has running time $O((t+1)\polylog(1/\eps)/\eps^4)$,
where $t$ is the average across the intervals of an upper bound on the time
needed to approximate $\|g-h\|_1^{[x_i,x_{i+1}]}$ to within $O(\eps^2/\log(1/\eps))$.
\end{theorem}
\begin{proof}
Let $f$ be a log-concave density such that  $\dtv(g, f)  = \opt$.
\new{First, we compute the median $\tilde \mu$ and interquartile range $\tilde \sigma$.
If $\opt \leq 1/10$, Lemma~\ref{lem:robust-moments} applies to these,
and otherwise Theorem~\ref{thm:main} is trivial.}
Using these approximations, we construct an interval $I \subseteq D$
containing at least $C\log(1/\epsilon)$ standard deviations about the mean of $f$,
and of total length $O(\log(1/\epsilon)\sigma)$, where $C$ is a sufficiently large constant.
If $D=\Z$ and $1/\eps=O(\sigma)$, we let $k$ be the length of $I$,
otherwise, we let $k=\Theta(\log(1/\eps)/\eps)$ and ensure that the length of $I$ divided by $k$ is in $D$.

We will attempt to find a log-concave pseudo-distribution $h$
satisfying the properties of Proposition~\ref{prop:this-form-close}
so that $\dtv(h,g)$ is (approximately) minimized.
Note that the proposition implies there exists a log-concave pseudo-distribution $h$ with $\dtv(f,h) = O(\eps)$,
and thus $\dtv(g,h)=O(\opt+ \eps)$. Given any such
$h$ with $\dtv(h,g)=O(\opt+\eps)$,  re-normalizing gives an explicit
log-concave density $h'$ with $\dtv(h',g)=O(\opt+\eps)$.

We find the best such $h$ via dynamic programming.
In particular, if $x_1,\ldots,x_{k+1}$ are the interval endpoints,
then $h$ is determined by the quantities $a_i=\log(h(x_i)\new{\tilde \sigma})$,
which are either $- m \cdot \eps/k$, where $m \in \Z_+$
$|m| \leq O((1/\eps) \log(1/\eps)k)$, or $-\infty$.
The condition that $h$ is log-concave
is equivalent to the sequence $a_i$ being concave.

\new{Let $S$ be the set of possible $a_i$ 's,
i.e., the multiples of $\epsilon/k$ in the range $[-\log(1/\eps), O(1)] \cup \{- \infty\}$.
Let $T$ be the set of possible $a_{i+1}-a_{i}$'s, i.e., numbers of the form $b \cdot \eps \cdot  2^c/k$,
for integers $|b| \leq (1/\eps) \log(1/\epsilon)$ and $0 \leq c \leq O(\log (1/\eps))$.
Let $H$ be the set of $h$ which satisfy the properties of
Proposition \ref{prop:this-form-close} except the bound on $\|h-f\|_1$.
}
We use dynamic programming to determine for each $i \in [k], a \in S, d \in T$
the concave sequence  $a_1,\ldots, a_i$ so that
$a_i=a$, $a_i-a_{i-1}=d$ and $\|h-g \|^{[x_1, x_i]}_1$
is as small as possible, where $h(x)$ is the density
obtained by interpolating the $a_i$'s by a piecewise exponential function.

We write $e_i(a_{i}, a_{i+1})$ for the error in the $i$-th interval $[x_{i}, x_{i+1}]$.
When $a_i$ and $a_{i+1}$ are both finite,
we take $e_i(a_{i},a_{i+1}) \eqdef  \|g-h_i\|_1^{[x_{i}, x_{i+1}]}$,
where $h_i(x) = \exp\left({\frac{x-x_i}{x_{i+1}-x_i}a_i + \frac{x_{i+1}-x_i}{x_{i+1}-x_i}a_{i+1}}\right)/\sigma.$
We define $e_i(a,-\infty)=e_i(-\infty,a) \eqdef \|g\|_1^{[x_{i}, x_{i+1}]}$.
Thus, we have  $\| g, h\|_1^{[x_1, x_i]} \leq \sum_{i=1}^{k} e_i(a_i,a_{i+1})$.
When $D=\R$, this an equality.
However, when $D = \Z$, we double count the error in the endpoints in the interior of the support,
and so have $\sum_{i=1}^{k} e_i(a_{i},a_{i+1}) \leq 2\| g-h\|_1^{[x_1, x_i]}.$

The algorithm computes $\tilde e_i(a_{i-1}, a_{i})$ with
$|\tilde e_i(a_{i-1}, a_{i}) - e_i(a_{i-1}, a_{i})| \leq \eps/k$ for all $a_{i-1},a_i \in S$
with $a_i-a_{i-1} \in T$.

\medskip

\fbox{\parbox{6in}{
{\bf Algorithm} {\tt Compute $h$}\\
Input: an oracle for computing $\tilde e_i(a,a')$

Output: a sequence $a_1, \ldots, a_n$ that minimizes $\sum_{i=1}^k e_{i+1}(a_i,a_{i+1})$


\begin{enumerate}
\item Let $G$ be the directed graph with vertices of the form:
\begin{enumerate}
\item $(0,-\infty,-)$,$(k+1,\infty,+)$
\item $(i,a,d)$ for $i\in[k], a\in S\backslash\{-\infty\},d\in T\cup\{\infty\}$
\item $(i,-\infty,s)$ for $i\in[k],s\in\{\pm \}$
\end{enumerate}
and weighted edges of the form
\begin{enumerate}
\item $(i,-\infty,-)$ to $(i+1,a,\infty)$ of weight $\tilde e_i(-\infty,a)$
\item $(i,a,d)$ to $(i+1,a+d,d)$ of weight $\tilde e_i(a,a+d)$
\item $(i,a,d)$ to $(i,a,d')$ with $d'$ the predecessor of $d$ in $T\cup\{\infty\}$ or weight $0$.
\item $(i,a,d)$ to $(i+1,-\infty,+)$ of weight $e_i(a,-\infty)$
\item $(i,-\infty,s)$ to $(i+1,-\infty,s)$ of weight $e_i(-\infty,-\infty)$
\end{enumerate}
\item Using the fact that $G$ is a DAG compute the path $P$ from $(0,-\infty,-)$ to $(k+1,-\infty,+)$ of smallest weight.
\item For each $i\in[k]$, let $a_i$ be the value such that $P$ passes through a vertex of the form $(i,a_i,d^{\ast})$.
\end{enumerate}
}}

\fbox{\parbox{6in}{
{\bf Algorithm} {\tt Full-Algorithm}\\
Input: A concise description of a distribution $g$ such that
$\dtv(f,g) \leq \opt$ for some log-concave distribution $f$ and $\eps > 0$.

Output: A log-concave continuous piecewise exponential $h$ with $\dtv(g,h) \leq O(\opt+\eps)$

\begin{enumerate}

\item Compute the median $\tilde \mu$ and interquartile range $\tilde \sigma$ of $g(x)$
\item If $D=\Z$ and $\tilde{\sigma}=O(1/\eps),$
\item \quad	let $\alpha=\tilde \mu-\Theta(\log(1/\eps)/\eps)$ and $\beta=\tilde \mu+\Theta(\log(1/\eps)/\eps)$ be integers,
$k=\beta-\alpha$ and $L=1,$
\item else let $L=\Theta(\eps \tilde\sigma)$ with $L \in \Z$ or $\R$ in the discrete and continuous cases respectively,
$k=\Theta(\log(1/\eps)/\eps)$ be an even integer, $\alpha=\tilde \mu - (k/2)$ and $\beta=\tilde \mu - (k/2)$.
\item Let $x_i=\alpha+(i+1)L$ for $1 \leq i \leq k+1$.
\item Let $S$ be the set of the multiples of $\epsilon/k$ in the range $[-\ln(1/\eps), O(1)] \cup \{- \infty\}$.
Let $T$ be the set of numbers of the form $b \cdot \eps \cdot  2^c/k$, for integers $|b| \leq (1/\eps) \ln(1/\epsilon)$ and $0 \leq c \leq O(\log (1/\eps))$.
\item Sort $T$ into ascending order.
\item Let $a_1, \ldots ,a_{k+1}$ be the output of algorithm {\tt Compute} $h$.
\item Return the continuous piecewise exponential $h(x)$ that has
$h(x_i)= \exp(a_i)/\tilde \sigma$ for all $1 \leq i \leq k+1$ and has endpoints
$x_l,x_{l+1},\ldots, x_r$, where $l$ and $r$ and the least and greatest $i$ such that $a_i$ is finite.

\end{enumerate}
}}

\medskip

Now we show correctness.
For every $h' \in H$, with log probabilities at the endpoints $a'_i$,
there is a path of weight $w_{h'} := \sum_{i=1}^{k} \tilde e_i(a'_i,a'_{i+1})$
which satisfies $$\|g-h'\|_1^I-\eps \leq w_{h'} \leq 2\|g-h'\|_1^I+\eps \;.$$
Thus, the output $h(x)$ has $\|g-h\|_1^I \leq 2\eps + 2\min_{h' \in H} \|g-h'\|_1^I$.
By Proposition~\ref{prop:this-form-close}, there is an $h^\ast \in H$
with $\|f-h^\ast\|_1^I \leq O(\eps)$, where $\dtv(f,g)=\opt$.
Thus, $\|g-h^\ast\|_1^I \leq \opt+O(\eps)$.
Therefore, we have that 
$$\|g-h\|_1^I \leq 2\|g-h^\ast\|_1^I + 2\eps \leq 2\opt+O(\eps) \;.$$

Since the mass of $f$ outside of $I$ is $O(\eps)$,
we have that the mass of $g$ outside of $I$ is at most $\opt+O(\eps)$.
Thus, $\dtv(g,h) \leq \opt + O(\eps) + \|g-h\|_1^I = O(\opt + \eps)$, as required.

Finally we analyze the time complexity.
The graph $G$ has $k|S||T|+2$ vertices.
Each vertex has at most one in-edge of each type.
Thus, we can find the shortest path in time $O(k|S||T|)$
plus the time it takes to compute every $\tilde e_i(a,a+d)$.
There are  $O(k|S||T|)$ such computations
and they take average time at most $t$.
Thus, the time complexity is
$$O(k|S||T|(t+1))=O(\log(1/\eps)/\eps \cdot \log(1/\eps)^2/\eps^2 \cdot \log(1/\eps)^2/\eps \cdot (t+1)) = O((t+1)\log^5(1/\eps)/\eps^4).$$
\end{proof}

\subsection{Putting Everything Together}

We are now ready to combine the various pieces that yield our main result.
Our starting point is the following non-proper learning algorithm:

\begin{theorem}[\cite{ADLS15}] \label{thm:agnostic-piecewise-linear}
There is an agnostic learning algorithm for $t$-piecewise linear distributions with sample complexity $O(t/\eps^2)$
and running time $O((t/\eps^2) \log(1/\eps))$. 
Moreover, the algorithm outputs an $O(t)$-piecewise linear hypothesis distribution.
\end{theorem}

To establish that our overall learning algorithm will have the optimal sample complexity of 
$O(\eps^{-5/2})$, we make use of the following approximation theorem:

\begin{theorem}\label{approxThrm}
For any log-concave density $f$ on either $\R$ or $\Z$,
and $\epsilon>0$, there exists a piecewise linear distribution
$g$  with $O(\eps^{-1/2})$ interval pieces so that $\dtv(f,g) \leq \eps$.
\end{theorem}

The proof of Theorem~\ref{approxThrm} is deferred to Appendix \ref{sec:thm4}.
We now have all the ingredients to prove our main result.

\begin{proof}[Proof of Theorem \ref{thm:main-final}]
Let $f'$ be a log-concave density with $\dtv(f,f')=\opt$.
By Theorem \ref{approxThrm}, there is a piecewise linear density $g'$
with $O(\eps^{-1/2})$ pieces that has $\dtv(f',g') \leq \opt+\eps$.
By Theorem  \ref{thm:agnostic-piecewise-linear}, there is an algorithm
with sample complexity $O(1/\eps^{5/2})$ and running time $O((1/\eps^{5/2}) \log(1/\eps))$
that computes a piecewise linear density $g'$ with $O(\eps^{-1/2})$ pieces
such that $\dtv(f',g') \leq O(\opt+\eps)$.
We apply the algorithm of Theorem \ref{thm:main} to this $g'$,
which produces a piecewise exponential approximation $h(x)$
that satisfies $\dtv(g',h) \leq O(\opt + \eps)$, and therefore
$\dtv(h, f) \leq O(\opt + \eps)$.

It remains to prove that the time complexity is $\tilde O(n^{8/5}) = \tilde O(1/\eps^4)$.
To obtain this, we must show that $t=\polylog(1/\eps)$ in the statement of Theorem \ref{thm:main}.
When $D=\Z$ and the length of each interval is $1$, we have $t=O(1)$.
Otherwise, we divide into $k=\Theta((1/\eps)\log(1/\eps))$ pieces.
Since $k \geq O(\eps^{-1/2})$, the average number of endpoints of $g'(x)$
in a piece of $h(x)$ is smaller than $1$. Thus, to get the amortized time
complexity to be $\polylog(1/\eps)$, it suffices to show this bound
for an exponential and linear function on a single interval.
The following claim is proved in Appendix~\ref{sec:claim2}:
\begin{claim} \label{claim:compute-error}
Let $g(x)=ax+b$ and $h(x)=c\exp(dx)$.
Let $I=[x',x'+L]$ be an interval with $g(x) \geq 0$
and $0 \leq h(x) \leq O(\eps/L)$ for all $x \in I$.
There is an algorithm which approximates
$\|g-h\|_1^I$ to within an additive $O(\eps^2/\log(1/\eps))$ in time $\polylog(1/\eps)$.
\end{claim}
This completes the proof of Theorem~\ref{thm:main-final}.
\end{proof}

\section{Discussion and Future Directions} \label{sec:concls}
In this paper, we gave the first agnostic learning algorithm for log-concave distributions
that runs in polynomial time. Our algorithm is sample-optimal and runs in time that is sub-quadratic
in the size of its input sample. The obvious open problem is to obtain an agnostic proper learning algorithm that runs
in near-linear time. More broadly, an interesting and challenging question is to generalize
our techniques to the problem of learning log-concave distributions in higher dimensions.

We believe that our algorithmic approach should naturally extend to other structured distribution families, e.g.,
to monotone hazard rate (MHR) distributions, but we have not pursued this direction. Finally, as we point out in the following
paragraph, our dynamic programming approach can be 
extended to properly learning mixtures of log-concave densities,
alas with running time exponential in the number $k$ of components, i.e., $(1/\eps)^{O(k)}$.  

Indeed, the non-proper learning algorithm from~\cite{ADLS15} also applies to mixtures, so it suffices 
to efficiently compute a nearly optimal approximation of a given distribution by a mixture of $k$ log-concave distributions.
It is easy to see that we can assume each of the mixing weights is $\Omega(\eps)$.
For our approach to work, we will need to approximate the mean and standard deviation of each distribution in the mixture.
This can be done if we have $O(1)$ samples from each component, 
which can be accomplished by taking $O(1)$ samples from our original distribution 
and noting that with probability $\Omega(\eps)^{O(1)}$ it has chosen only samples from the desired component. 
After doing this, we will need to build a larger dynamic program.
Our new dynamic program will attempt to approximate $f$ as a mixture of functions $h$ of the form 
given in Proposition~\ref{prop:this-form-close}. Specifically, it will need to have steps corresponding to each of the $x_i$'s
for each of the functions $h$, and will need to keep track of both the current value of each $h$
and its current logarithmic derivative. 

The aforementioned discussion naturally leads to our final open problem:
Is there a proper learning algorithm for mixtures of $k$ log-concave distributions 
with running time $\poly(k/\eps)$?


\bibliographystyle{alpha}
\bibliography{allrefs}

\appendix
\section{Omitted Proofs}




\subsection{Proof of Lemma \ref{lem:robust-moments}} \label{sec:app-lemma2}

\medskip

\noindent {\bf Lemma \ref{lem:robust-moments}}
{\em
Let $f: D \to \R_+$ be a log-concave density with mean $\mu$ and standard deviation $\sigma$. 
\new{Assume that either $D=\R$ or that $\sigma$ is sufficiently large.}
Let $g: D \to \R_+$ be a density with $\dtv(f, g) \leq 1/10$.
Given an explicit description of $g$, we can efficiently compute values $\tilde \mu$ and $\tilde \sigma$ so that
$|\mu-\tilde\mu| \leq 2\sigma$ and $3\sigma/10 \leq \tilde \sigma \leq 6\sigma.$
}

\medskip

\begin{proof}
We define $\tilde \mu$ to be the median of $g$
and $\tilde \sigma$ to be the difference between the $25^{th}$ and $75^{th}$ percentiles of $g$.
Since $f$ and $g$ are within total variation distance $1/10$, it follows that their Kolmogorov distance
(i.e., the maximum distance between their cumulative distribution functions) is at most $1/10$.
This implies that $\tilde \mu$ lies between the $40^{th}$ and $60^{th}$ percentiles of $f$.
By Cantelli's inequality, we have that $\Pr_{X \sim f}[X - \mu \geq 2 \sigma] \leq 1/5$
and $\Pr_{X \sim f}[X - \mu \leq -2 \sigma] \leq 1/5$. Thus, $|\mu - \tilde \mu| \leq 2 \sigma$.

Similarly, $\tilde \sigma$ lies between (a) the difference between the $65^{th}$ and $35^{th}$ percentile of $f$
and (b) the difference between the $85^{th}$ and $15^{th}$ percentile of $f$.
By Cantelli's inequality, we have that
$\Pr_{X \sim f}[X - \mu \geq 3 \sigma] \leq 1/10$ and
$\Pr_{X \sim f}[X - \mu \leq -3 \sigma] \leq 1/10$.
Thus, $\tilde \sigma \leq 6 \sigma$. For the other direction, note that
$3/10$ of the probability mass of $f$ lies between the $35^{th}$ and $65^{th}$ percentile.
Since the maximum value of $f$ is at most $1/\sigma$, by Lemma~\ref{lem:lc-facts} (i),
we conclude that the difference between the $65^{th}$ and $35^{th}$ percentiles is at least $3\sigma/10$.
\end{proof}

\subsection{Proof of Claim \ref{clm:tail-bound}} \label{sec:app-claim1}
\medskip

\noindent {\bf Claim \ref{clm:tail-bound}}
{\em
Let $f$ be a log-concave density with mean $\mu$ and standard deviation $\sigma$.
Let $\alpha \leq \mu - \Omega(\sigma(1+ \log(1/\eps)))$ and
$\beta \geq \mu + \Omega(\sigma(1+ \log(1/\eps))))$.
Then, $\|f\|_1^{(-\infty,\alpha)} \leq \eps$ and $\|f\|_1^{(\beta,\infty)} \leq \eps.$
}
\medskip

\begin{proof}
By Lemma~\ref{lem:lc-facts} (ii), we have
$f(x) \leq \exp\left(1-\frac{|x-\mu|}{8e\sigma}\right)/\sigma$.
In the case $D = \R$, we have
$\int_{-\infty}^{\alpha} f(x) dx \leq \int_{-\infty}^{\alpha} \exp\left(1-\frac{|x-\mu|}{8e\sigma}\right)/\sigma dx
= 8e \sigma \exp\left(1-\frac{|{\alpha}-\mu|}{8e\sigma}\right)$.
This is at most $\eps$ when $|{\alpha}-\mu| \geq 8e \sigma(1+ \ln(8e/\eps))$, 
which holds by our bounds on $|\alpha-\mu|$.

In the case $D = \Z$, we have
$\sum_{-\infty}^{{\alpha}-1} f(x) \leq \sum_{-\infty}^{{\alpha}-1} \exp\left(1-\frac{|x-\mu|}{8e\sigma}\right)/\sigma$.
Since $\exp\left(1-\frac{|x-\mu|}{8e\sigma}\right)/\sigma$ is monotonically increasing on $(\infty,{\alpha}]$,
we have that $\exp\left(1-\frac{|x-\mu|}{8e\sigma}\right)/\sigma \leq \int_x^{x+1} \exp\left(1-\frac{|y-\mu|}{8e\sigma}\right)/\sigma dy$
for $x \leq {\alpha}-1$.
Thus, we can bound this sum by the same integral to the one for the continuous case, i.e.,
$
\sum_{-\infty}^{{\alpha}-1} \exp\left(1-\frac{|x-\mu|}{8e\sigma}\right)/\sigma
\leq \int_{-\infty}^{{\alpha}} \exp\left(1-\frac{|x-\mu|}{8e\sigma}\right)/\sigma dx \leq \eps \;.
$
A symmetric argument yields that the probability mass of $f$ on $(\beta,\infty)$ is $O(\eps)$.
\end{proof}

\subsection{Proof of Theorem~\ref{approxThrm}} \label{sec:thm4}
\medskip

\noindent {\bf Theorem \ref{approxThrm}}
{\em
If $f$ is a log-concave density on either $\R$ or $\Z$,
and $\epsilon>0$, there exists a piecewise linear distribution
$g$ with $O(\eps^{-1/2})$ interval pieces so that $\dtv(f,g) \leq \eps$.
}
\medskip

We begin by proving this in the case where the range of $f$ and the logarithmic derivative of $f$ are both relatively small.

\begin{lemma}\label{piecewiseApproxLem}
Let $f$ be a log-concave function defined on an interval $I$ in either $\R$ or $\Z$.
Suppose furthermore, that the range of $f$ is contained in an interval of the form $[a,2a]$ for some $a$,
and that the logarithmic derivative of $f$ (or the log-finite difference of $f$ in the discrete case)
varies by at most $1/|I|$ on $I$. Then there exists a piecewise linear function $g$
on $I$ with $O(\eps^{-1/2})$ pieces so that $\|f-g\|_1 \leq O(\eps\|f\|_1).$
\end{lemma}
\begin{proof}
By scaling $f$, we may assume that $a=1$. Note that the log-derivative or log-finite difference of $f$ must be $O(1/|I|)$ everywhere. We now partition $I$ into subintervals $J_1,J_2,\ldots,J_n$ so that on each $J_i$ has length at most $\eps^{1/2}|I|$ and the logarithmic derivative (or finite difference) varies by at most $\eps^{1/2}/|I|$. Note that this can be achieved with $n=O(\eps^{-1/2})$ by placing an interval boundary every $O(\eps^{1/2}|I|)$ distance as well as every time the logarithmic derivative passes a multiple of $\eps^{1/2}/|I|$.

We now claim that on each interval $J_i$ there exists a linear function $g_i$ so that $\|g_i-f\|_\infty = O(\eps)$. Letting $g$ be $g_i$ on $J_i$ will complete the proof.

Let $J_i = [y,z]$. We note that for $x\in J_i$ that
$$
f(x) = f(y)\exp((x-y)\alpha)
$$
for some $\alpha$ in the range spanned by the logarithmic derivative (or log finite difference) of $f$ on $J_i$. Letting $\alpha_0$ be some number in this range, we have that
\begin{align*}
f(x)&=f(y)\exp((x-y)\alpha_0 + (x-y)(\alpha-\alpha_0))\\
& =f(y)\exp((x-y)\alpha_0)\exp(O(\eps^{1/2}|I|)O(\eps^{1/2}/|I|))\\
&=(1+O(\eps))f(y)\exp((x-y)\alpha_0) \;.
\end{align*}
Noting that $(x-y)\alpha_0 = O(\eps^{1/2}|I|)O(1/|I|)=O(\eps^{1/2}),$ this is
$$
(1+O(\eps))f(y)(1+(x-y)\alpha_0+O((x-y)\alpha_0)^2)=(1+O(\eps))(f(y)+(x-y)\alpha_0+O(\eps)) = f(y)+(x-y)\alpha_0+O(\eps).
$$
Therefore, taking $g_i(x) = f(y)+(x-y)\alpha_0$ suffices. This completes the proof.
\end{proof}

Next, we need to show that we can partition the domain of $f$ into intervals $I$ satisfying the above properties.
\begin{proposition}\label{intervalDivisionProp}
Let $f$ be a log-concave distribution on either $\R$ or $\Z$. Then there exists a partition of $\R$ or $\Z$ into disjoint intervals $I_1,I_2,\ldots$ so that
\begin{itemize}
\item $f$ satisfies the hypotheses of Lemma \ref{piecewiseApproxLem} on each $I_i$.
\item For each $m$, there are only $O(m)$ values of $i$ so that $f(I_i) > 2^{-m}$.
\end{itemize}
\end{proposition}
\begin{proof}
Firstly, by splitting the domain of $f$ into two pieces separated by the modal value, we may assume that $f$ is monotonic. Henceforth, we assume that $f$ is defined on $\R^+$ or $\Z^+$ and that $f$ is both log-concave and monotonically decreasing.

We define the intervals $I_i=[a_i,b_i]$ inductively. We let $a_1=0$. Given $a_i$, we let $b_i$ be the largest possible value so that $f$ restricted to $[a_i,b_i]$ satisfies the hypotheses of Lemma \ref{piecewiseApproxLem}. Given $b_i$ we let $a_{i+1}$ be either $b_i$ (in the continuous case) or $b_i+1$ (in the discrete case). Note that this causes the first condition to hold automatically.

We note that for each $i$, either $f(a_{i+1})\leq f(a_i)/2$ or the logarithmic derivative of $f$ at $a_{i+1}$ is less than the logarithmic derivative at $a_i$ by at least $1/(a_{i+1}-a_i)$. Note that in the latter case, since $f(a_{i+1})>f(a_i)/2$, we have that the absolute value of the logarithmic derivative at $f(a_i)$ is at most $O(1/(a_{i+1}-a_i))$. Therefore, in this latter case, the absolute value of the logarithmic derivative of $f$ at $a_{i+1}$ is larger than the absolute value of the logarithmic derivative at $a_i$ by at least a constant multiple.

Note that at the end of the first interval, we have that $f(a_2) = O(1/|I_1|)$ and that the absolute logarithmic derivative of $f$ at $a_2$ is at least $\Omega(1/|I_1|)$. Note that each interval at least one of these increases by a constant multiple, therefore, there are only $O(m)$ many $i$ so that both $f(a_i) > 2^{-m}/|I_1|$ and the absolute logarithmic derivative of $f$ at $a_i$ is less than $2^m/|I_1|$. We claim that if either of these fail to hold that the integral of $f$ over $I_i$ is $O(2^{-m})$.

If $f(a_i)<2^{-m}/|I_1|$, then since the absolute logarithmic derivative of $f$ on $I_i$ is at least $\Omega(1/|I_1|)$, we have that the length of $I_i$ is $O(|I_1|)$. Therefore, the mass of $f$ on $I_i$ is $O(2^{-m}).$

If on the other hand the absolute logarithmic derivative of $f$ at $a_i$ is at least $2^m/|I_1|$, since the value if $f$ on $I_i$ varies by at most a multiple of $2$, we have that $|I_i| = O(|I_1|/2^m)$. Since $f$ is decreasing, is has size $O(1/|I_1|)$ on $I_i$, and therefore, the integral of $f$ of $I_i$ is $O(2^{-m})$. This completes the proof of the second condition.
\end{proof}

We are now prepared to prove our Theorem~\ref{approxThrm}:

\begin{proof}
We divide $\R$ or $\Z$ into intervals as described in Proposition \ref{intervalDivisionProp}. Call these intervals $I_1,I_2,\ldots$ sorted so that $f(I_i)$ is decreasing in $i$. Therefore, we have that $f(I_m) = O(2^{-\Omega(m)}).$ In particular, there is a constant $c>0$ so that $f(I_m)=O(2^{-cm})$.

For $m=1,\ldots,2\log(1/\eps)/c$, we use Lemma \ref{piecewiseApproxLem} to approximate $f$ in $I_m$ by a piecewise linear function $g_m$ so that $g_m$ has at most $O(\eps^{-1/2}2^{-cm/4})$ pieces and so that the $L_1$ distance between $f$ and $g_m$ on $I_m$ is at most $f(I_m)O(\eps2^{cm/2})= O(\eps2^{-cm/2}).$ Let $g$ be the piecewise linear function that is $g_m$ on $I_m$ for $m\leq c\log(1/\eps)/2$, and $0$ elsewhere. $g$ is piecewise linear on
$$
\sum_{m=1}^{2\log(1/\eps)/c} O(\eps^{-1/2}2^{-cm/4}) = O(\eps^{-1/2})
$$
intervals.

Furthermore the $L_1$ error between $f$ and $g$ on the $I_m$ with $m\leq 2\log(1/\eps)/c$ is at most
$$
\sum_{m=1}^{2\log(1/\eps)/c} O(\eps2^{-cm/2}) = O(\eps).
$$

The $L_1$ error from other intervals is at most
$$
\sum_{m=2\log(1/\eps)/c}^\infty O(2^{-cm}) = O(\eps).
$$
Therefore, $\|f-g\|_1=O(\eps)$.

By replacing $g$ by $\max(g,0)$, we may ensure that it is positive (and at most double the number of pieces and decrease the distance from $f$). By scaling $g$, we may then ensure that it is a distribution. Finally by decreasing $\eps$ by an appropriate constant, we may ensure that $\dtv(f,g)\leq \eps.$ This completes the proof.
\end{proof}

\subsection{Proof of Claim~\ref{claim:compute-error}} \label{sec:claim2}
\medskip

\noindent {\bf Claim \ref{claim:compute-error}}
{\em Let $g(x)=ax+b$ and $h(x)=c\exp(dx)$.
Let $I=[x', x'+L]$ be an interval with $g(x) \geq 0$
and $0 \leq h(x) \leq O(\eps/L)$ for all $x \in I$.
There is an algorithm which approximates
$\|g-h\|_1^I$ to within an additive $O(\eps^2/\log(1/\eps))$ in time $\polylog(1/\eps)$.}

\medskip

\begin{proof}
First, we claim that for any subinterval $I' \subseteq I$,
we can compute $\|h\|_1^{I'}$ and $\|g\|_1^{I'}$
to within $O(\eps^2/\log(1/\eps))$ in time $\polylog(1/\eps)$.
There are simple closed formulas for the integrals of these
and for the sum of arithmetic and geometric series.
The formula for the sum of a geometric series may
have a cancellation issue when the denominator $1-\exp(d)$
is small but note that when $|d|=O(\eps^2/\log(1/\eps))$
we can approximate the sum of $h(x)$ over $I$ by its integral.
These can all be computed in $\polylog(1/\eps)$ time.

Now it remains to approximate any crossing points,
i.e., points $x$ where $g(x)=h(x)$ for $x' \leq x \leq x'+L$
(which need not satisfy $x \in D$).
If we find these with sufficient precision,
then we can divide $I$ and calculate
$\|h\|_1^{I'}$ and $\|g\|_1^{I'}$ for each sub-interval $I'$
to get the result. Note that $g(x)$ and $h(x)$
can have at most two crossing points, since there is at most one
$x \in \R$ where the derivative of $h(x)-g(x)$ is $0$.
This can be calculated as $x^{\ast}= \ln(a/(cd))/d$, when $a/(cd) > 0$.
If $x^{\ast}$ lies in $I$ we can subdivide and reduce to the case
when there is a crossing point only if $g(x)-h(x)$
has different signs at the endpoints. In this case,
if $g(x) = \Omega(\eps/L)$ at one endpoint,
we can find a point at which $g(x)=\Theta(\eps/L)$,
which is higher than our bound on $h(x)$, and we can divide there.

Thus, we can reduce to the case where there is exactly
one crossing point in $I$ and $h(x), g(x) \leq O(\eps/L)$.
By performing $O(\log(1/\eps))$ bisections we can approximate
this crossing point to within $O(\eps L/\log(1/\eps)$
(or $\max \{1,O(\eps^2 L/\log(1/\eps)\}$ when $D=\Z$).
Then, we have that if $J$ is the interval between the true
crossing point and our estimate, then
$\|g-h\|_1^J \leq \|g\|_1^J+\|h\|_1^J = O(L\eps/\log(1/\eps) \cdot \eps/L) = O(\eps^2/\log(1/\eps)$.
Hence, if we divide here, each of the sub-intervals $I'$ has
$\|h-g\|_1^{I'}=|\|h\|_1^{I'} - \|g\|_1^{I'}| + O(\eps^2/\log(1/\eps))$.

Note that in all of the above cases, we only sub-divide $I$ into at most $O(1)$
sub-intervals, and so it takes $\polylog(1/\eps)$ time to compute $\|g-h\|_1^I$ for the whole interval.
\end{proof}

\end{document}